\documentclass[conference,letterpaper]{IEEEtran}
\usepackage[utf8]{inputenc} 
\usepackage[T1]{fontenc}
\usepackage{ifthen}
\usepackage[cmex10]{amsmath} 

\usepackage{mathtools,amssymb,lipsum}

\usepackage{cuted}
\usepackage{tikz}
\usetikzlibrary{arrows,chains,matrix,positioning,scopes,shapes,decorations.pathreplacing}
\pgfdeclarelayer{background}
\pgfdeclarelayer{behind}
\pgfdeclarelayer{above}
\pgfdeclarelayer{glass}
\pgfsetlayers{background,behind,main,above,glass}
\newlength{\mywidth}
\makeatletter
\tikzset{join/.code=\tikzset{after node path={%
\ifx\tikzchainprevious\pgfutil@empty\else(\tikzchainprevious)%
edge[every join]#1(\tikzchaincurrent)\fi}}}
\makeatother
\tikzset{>=stealth',every on chain/.append style={join},
         every join/.style={->}}
\tikzstyle{labeled}=[execute at begin node=$\scriptstyle,
   execute at end node=$]
\makeatletter
\def\ps@headings{%
\def\@oddhead{\mbox{}\scriptsize\rightmark \hfil \thepage}%
\def\@evenhead{\scriptsize\thepage \hfil \leftmark\mbox{}}%
\def\@oddfoot{}%
\def\@evenfoot{}}
\makeatother 
\IEEEoverridecommandlockouts
\usepackage{balance} 
\usepackage{color}
\usepackage{amsfonts}
\usepackage{caption}
\usepackage{subfigure}
\usepackage{tikz}
\usepackage{times}
\usepackage{cite}
\usepackage{lettrine}
\usepackage{amsthm}

\allowdisplaybreaks[4]
\usepackage{array} 
\usepackage{amssymb}

\usepackage{stfloats}
\usepackage{graphicx}
\usepackage{footnote}
\usepackage{booktabs}
\usepackage{array}
\usepackage{algorithmic}
\usepackage{algorithm}
\usepackage{subeqnarray}
\usepackage{cases}
\usepackage{threeparttable}
\usepackage{color}
\usepackage{url}
\usepackage{bm}
\usepackage{physics}
\usepackage{hyperref}
\newtheorem{theorem}{Theorem}

\newtheorem{corollary}[theorem]{Corollary}
\newtheorem{proposition}[theorem]{Proposition}
\newtheorem{example}[theorem]{Example}

\newtheorem{definition}[theorem]{Definition}

\newcommand{\setE}{ \mathcal{E} }
\newcommand{\setvar}{ \mathcal{E}(\Nfcyc) }
\newcommand{\setF}{ \mathcal{F} }
\newcommand{\beli}{\beta}
\newcommand{\vbeli}{ \bm{\beli} }

\newcommand{\tran}{ \mathsf{T} }
\newcommand{\herm}{ \mathsf{H} }
\newcommand{\cov}{ \mathrm{Cov} }
\newcommand{\corr}{ \mathrm{Corr} }

\newcommand{\vari}{ \mathrm{Var} }
\newcommand{\MP}{ \mathcal{M} }

\newcommand{\rCHSHbet}{ \mathrm{CorrCHSH}(\vbeli) }
\newcommand{\Nfcyc}{ \mathsf{N}_{1} } 
\newcommand{\NMkov}{ \mathsf{N}_{2} } 
\newcommand{\SNFG}{ \mathsf{N}_{3} } 
\newcommand{\NQFG}{ \mathsf{N}_{4} } 
 
\newcommand{\pr}{ \mathrm{Pr} }
\newcommand{\pn}{ p_{ \mathsf{N} } }

\newcommand{\vxf}{ \bm{x}_{\partial f} }
\newcommand{\vpf}{ p_{\mathsf{N}, i,j } }

\newcommand{\vpe}{ p_{\mathsf{N}, i } }

\newcommand{\tx}{ \tilde{x} }

\newcommand{\tvx}{ \tilde{\bm{x}} }
\newcommand{\tzero}{ \tilde{0} }
\newcommand{\tone}{ \tilde{1} }
\newcommand{\setx}{ \mathcal{X} }

\newcommand{\setxe}{ \setx_{e} }
\newcommand{\sety}{ \mathcal{Y} }

\newcommand{\vx}{ \bm{x} }
\newcommand{\vxI}{ \bm{x}_{\set{I}} }
\newcommand{\vxIc}{ \bm{x}_{\set{I}^{\mathrm{c}}} }
\newcommand{\vY}{ \bm{Y} }
\newcommand{\vy}{ \bm{y} }

\newcommand{\vq}{ \bm{q} }

\newcommand{\setL}{ \mathcal{LM}(\set{K}) }
\newcommand{\setLCHSH}{ \mathcal{LM}_{\mathrm{CHSH}}(\set{K}) }
\newcommand{\sR}{\mathbb{R}}
\newcommand{\sC}{\mathbb{C}}
\newcommand{\set}[1]{\mathcal{#1}}

\begin{document}
\bibliographystyle{IEEEtran}

\title{Sets of Marginals and Pearson-Correlation-based CHSH Inequalities for a Two-Qubit System}

\IEEEoverridecommandlockouts

\author{\mbox{} \\[-0.60cm]
  \IEEEauthorblockN{Yuwen~Huang and Pascal O.\ Vontobel}
  \IEEEauthorblockA{Department of Information Engineering \\
                    The Chinese University of Hong Kong \\
                    hy018@ie.cuhk.edu.hk, pascal.vontobel@ieee.org \\[-0.25cm]}
  \thanks{This work has been supported in part by the Research Grants Council of the Hong Kong Special Administrative Region, China, under Project CUHK 14209317 and Project CUHK 14207518.}
}
\maketitle
\IEEEpeerreviewmaketitle

\begin{abstract}
  Quantum mass functions (QMFs), which are tightly related to decoherence
  functionals, were introduced by Loeliger and Vontobel [IEEE Trans.\ Inf.\
  Theory, 2017, 2020] as a generalization of probability mass functions
  toward modeling quantum information processing setups in terms of factor
  graphs. 

  Simple quantum mass functions (SQMFs) are a special class of QMFs that do
  \emph{not explicitly} model classical random variables. Nevertheless,
  classical random variables appear \emph{implicitly} in an SQMF if some
  marginals of the SQMF satisfy some conditions; variables of the SQMF
  corresponding to these ``emerging'' random variables are called classicable
  variables. Of particular interest are jointly classicable variables.

  In this paper we initiate the characterization of the set of marginals given
  by the collection of jointly classicable variables of a graphical model and
  compare them with other concepts associated with graphical models like the
  sets of realizable marginals and the local marginal polytope.

  In order to further characterize this set of marginals given by the
  collection of jointly classicable variables, we generalize the CHSH
  inequality based on the Pearson correlation coefficients, and thereby prove
  a conjecture proposed by Pozsgay \textit{et al}. A crucial feature of this inequality
  is its nonlinearity, which poses difficulties in the proof.

%

\end{abstract}

\section{Introduction}


Graphical models like factor graphs~\cite{Kschischang2001, Forney2001,
  Loeliger2004a} have been used to represent various statistical models. In
the following, we will call a factor graph consisting only of non-negative
real-valued local functions a standard factor graph (S-FG). S-FGs have many
applications, in particular in communications and coding theory (see, e.g.,
\cite{Wymeersch2007, Richardson2008}) and statistical mechanics (see, e.g.,
\cite{Mezard2009}). In these applications, factor graphs frequently represent
the factorization of the joint probability mass functions (PMFs) of all the
relevant random variables. Quantities of interest can then be obtained by
exactly or approximately computing marginals of this joint PMF and suitably
processing these marginals.


Factor graphs have also been used to represent quantum-mechanical
probabilities~\cite{Loeliger2017, Loeliger2020}. In contrast to S-FGs, these
factor graphs consist of complex-valued local functions satisfying some
constraints. In the following, we will call such factor graphs
quantum-probability factor graphs (Q-FGs). A Q-FG is typically used to
represent the factorization of the joint quantum mass function (QMF) as
introduced in~\cite{Loeliger2017}.


In this paper, we first discuss similarities and differences between PMFs and
QMFs. Some of the features of QMFs will then motivate the study that is
carried out in the rest of this paper.


\section{PMFs vs.\ QMFs}


In this section, we highlight some similarities and crucial differences between
PMFs and QMFs. First, we consider a classical setup. In particular, we assume
that we are interested in a graphical model that represents the joint PMF
$P_{Y_1, \ldots, Y_n}(y_1, \ldots, y_n)$, where $Y_1, \ldots, Y_n$ are some
random variables of interest taking value in some alphabets
$\sety_1, \ldots, \sety_n$.\footnote{For simplicity, in the following all
  alphabets will be finite.} (In a typical application, we might have observed
$Y_1 = y_1, \ldots, Y_{n-1} = y_{n-1}$ and would like to estimate $Y_n$ based
on these observations.) In most applications, the PMF
$P_{Y_1, \ldots, Y_n}(y_1, \ldots, y_n)$ does \emph{not} have a ``nice''
factorization in terms of simple factors. However, frequently, with the
introduction of suitable auxiliary variables $x_1, \ldots, x_m$ taking values
in some alphabets $\setx_1, \ldots, \setx_m$, respectively, there is a function
$p(\vx, \vy)$, where $\vx := (x_1, \ldots, x_m)$ and
$\vy := (y_1, \ldots, y_n)$, such that
\begin{align*}
  p(\vx,\vy) 
    &\in \sR_{\geq 0} \quad \text{(for all $\vx, \vy$)} \ , \\
  \quad 
  \sum_{\vx,\vy} 
    p(\vx,\vy)
    &= 1 \ , \\
  \sum_{\vx} 
    p(\vx,\vy)
    &= P_{\vY}(\vy) \quad \text{(for all $\vy$)} \ ,
\end{align*}
and such that $p(\vx,\vy)$ has a ``nice'' factorization. (For example, in a
hidden Markov model, the joint PMF of the observations does not have a
``nice'' factorization, but the joint PMF of the hidden state process and the
observations has a ``nice'' factorization.) Note that the function
$p(\vx, \vy)$ can, thanks to its properties, be considered as a joint PMF of some
random variables $X_1, \ldots, X_m, Y_1, \ldots, Y_n$.


Second, we consider a quantum-mechanical setup. We assume, again, that we are
interested in a graphical model representing the joint PMF
$P_{Y_1, \ldots, Y_n}(y_1, \ldots, y_n)$, where $Y_1, \ldots, Y_n$ are some
random variables of interest taking values in some alphabets
$\sety_1, \ldots, \sety_n$. Such random variables can, for example, represent
the measurements obtained when running some quantum-mechanical experiment, and
we might be interested in estimating $Y_n$ based on the observations
$Y_1 = y_1, \ldots, Y_{n-1} = y_{n-1}$. As in the classical case, the PMF
$P_{Y_1, \ldots, Y_n}(y_1, \ldots, y_n)$ usually does \emph{not} have a
``nice'' factorization in terms of simple factors. Moreover, standard physical
modeling of quantum-mechanical systems shows that introducing a function
$p(\vx,\vy)$ as defined above does usually \emph{not} help toward obtaining a
function with a ``nice'' factorization. However, in many quantum-mechanical
setups of interest, with the introduction of suitable auxiliary variables
$x_1, \ldots, x_m, x'_1, \ldots, x'_m$ taking values in some alphabets
$\setx_1, \ldots, \setx_m$, $\setx'_1, \ldots, \setx'_m$ (with
$\setx'_i = \setx_i$, $i \in \{ 1, \ldots, m \}$), there is a function
$q(\vx,\vx',\vy)$, called quantum mass function (QMF)~\cite{Loeliger2017},
such that
\begin{align*}
  q(\vx,\vx',\vy) 
    &\in \sC \quad \text{(for all $\vx, \vx', \vy$)} \ , \\
  \quad 
  \sum_{\vx,\vx',\vy} 
    q(\vx,\vx',\vy)
    &= 1 \ , \\
  q(\vx,\vx',\vy)
    &\ \text{is a PSD kernel in $(\vx,\vx')$ for every $\vy$} \ , \\
  \sum_{\vx,\vx'} 
    q(\vx,\vx',\vy)
    &= P_{\vY}(\vy) \quad \text{(for all $\vy$)} \ ,
\end{align*}
and such that $q(\vx,\vx',\vy)$ has a ``nice'' factorization. The major
difference between $p(\vx,\vy)$ and $q(\vx,\vx',\vy)$ is the fact that the
former takes value in $\sR_{\geq 0}$, whereas the latter takes value in
$\sC$. In particular, $\sum_{\vy} q(\vx,\vx',\vy)$ is in general not a PMF
over $(\vx, \vx')$, thereby showing that $\vx,\vx'$ \emph{cannot} be
considered as random variables. (See~\cite{Loeliger2017} for more details.)


In~\cite{Loeliger2020}, the authors discussed an approach to QMFs where $\vy$
does \emph{not} appear explicitly anymore, but ``emerges'' from a QMF. More
precisely, they first introduced a simple quantum mass function (SQMF)
$q(\vx,\vx')$ that satisfies
\begin{align*}
  q(\vx,\vx') 
    &\in \sC_{\geq 0} \quad \text{(for all $\vx, \vx'$)} \ , \\
  \quad 
  \sum_{\vx,\vx'} 
    q(\vx,\vx')
    &= 1 \ , \\
  q(\vx,\vx')
    &\ \text{is a PSD kernel in $(\vx,\vx')$} \ .
\end{align*}
Afterwards, they defined ``classicable'' variables.


\begin{definition}\label{Sec0:def:3}
  Let $\set{I}$ be a subset of $\{ 1, \ldots, m \}$ and let
  $\set{I}^{\mathrm{c}} := \{ 1, \ldots, m \} \setminus \set{I}$ be its
  complement. The variables $\vx_{\set{I}}$ are called jointly classicable if
  the function
  \begin{align*}
    q(\vxI,\vxI') 
      &:= \sum_{\vxIc,\vxIc'} 
            q(\vx,\vx')
  \end{align*}
  is zero for all $(\vxI,\vxI')$ satisfying $\vxI \neq \vxI'$.\footnote{It
    would be more precise to call this function $q_{\set{I}}$. However, for
    conciseness, we drop the index $\set{I}$ as it can be inferred from the
    arguments.}
\end{definition}


Note that if $\vxI$ are jointly classicable, then one can define the function
$p(\vxI) := q(\vxI,\vxI)$, for which it is straightforward, thanks to the
properties of SQMFs, to show that it is a PMF. It is in this sense that random
variables $y_1, \ldots, y_n$ that were omitted when going from QMFs to SQMFs
can ``emerge'' again.\footnote{Note that there is a strong connection of SQMFs
  to the so-called decoherence functional~\cite{GellMann1989, Dowker1992}, and
  via this also to the consistent-histories approach to quantum
  mechanics~\cite{Griffiths2002}. However, the starting point of our
  investigations is quite different.}


\begin{definition}
  Let $\set{K}$ be a collection of subsets $\set{I}$ of
  $\{ 1, \ldots, m \}$ such that $\vxI$ is classicable.
\end{definition}


\begin{example}
  Consider the Q-FG $\NQFG$ in
  Fig.~\ref{Sec1:fig:1}, whose global function is an SQMF. In that Q-FG,
  the matrix $\rho$ represents a PSD matrix and the matrices $U_1$, $U_2$ are unitary matrices. One can show that for all choices of $\rho$, $U_1$, and $U_2$, the collection
  $\set{K}$ can be chosen to contain the sets $\{ 1, 2 \}$, $\{ 1, 4 \}$,
  $\{ 2, 3 \}$, and $\{ 3, 4 \}$.
\end{example}


Interestingly enough, the collection of functions
$\bigl\{ p(\vxI) \bigr\}_{\set{I} \in \set{K}}$ is usually such that there is
\emph{no} PMF $p(\vx)$ such that for every $\set{I} \in \set{K}$, the function
$p(\vxI)$ can be obtained as a marginal of $p(\vx)$.\footnote{A similar
  observation is at the origin of the so-called ``single-framework'' rule in
  the consistent-histories approach to quantum mechanics.} In general, we can
only guarantee that for two sets $\set{I}_1, \ \set{I}_2 \in \set{K}$ the
following consistency constraint holds:
\begin{align*}
  \sum_{\vx_{\set{I}_1 \setminus \set{I}_2}} p(\vx_{\set{I}_1})
    &= \sum_{\vx_{\set{I}_2 \setminus \set{I}_1}} p(\vx_{\set{I}_2})
         \quad \text{(for all $\vx_{\set{I}_1 \cap \set{I}_2}$)} \ .
\end{align*}


Let us comment on these special properties of
$\bigl\{ p(\vxI) \bigr\}_{\set{I} \in \set{K}}$:
\begin{itemize}

\item It turns out that these special properties of
  $\bigl\{ p(\vxI) \bigr\}_{\set{I} \in \set{K}}$ are at the heart of quantum
  mechanical phenomena like Hardy's paradox~\cite{Hardy1992} and the
  Frauchiger--Renner paradox~\cite{Frauchiger2018}.\footnote{For a discussion
    of the latter in terms of SQMFs, see~\cite{Loeliger2020}.} In fact, the
  Q-FG $\NQFG$ in Fig.~\ref{Sec1:fig:1} can be used to analyze Hardy's
  paradox. On the side, note that $\NQFG$ also captures the essence of Bell's game~\cite{Gisin2014}.

\item Interestingly, these special properties of
  $\bigl\{ p(\vxI) \bigr\}_{\set{I} \in \set{K}}$ are very similar to the
  properties of the beliefs in the local marginal polytope of an S-FG~(see, e.g.,
  \cite{Wainwright2008}).\footnote{Local marginal polytopes are of relevance,
    for example, when characterizing locally operating message-passing
    iterative algorithms like the sum-product algorithm~\cite{Yedidia2005,
      Vontobel2013}.}

\end{itemize}


The above observations motivate the systematic study of the collection
$\bigl\{ p(\vxI) \bigr\}_{\set{I} \in \set{K}}$ for a given SQMF. Indeed, one
key contribution of this paper is to study this collection for the Q-FG 
$\NQFG$ in Fig.~\ref{Sec1:fig:1} and compare this collection with other 
objects that can be associated with this Q-FG.


\section{Contributions}


\begin{figure}[t]
	\centering
	\begin{minipage}[b]{0.15\textwidth}
		\centering
		\begin{tikzpicture}[node distance=1.5cm, on grid,auto]
  \tikzstyle{state}=[shape=rectangle, draw, minimum width=0.4cm, minimum height = 0.4cm,outer sep=-0.3pt]
  \begin{pgfonlayer}{main}
    \node[state] (f0) at (0,0) [label=above:\scriptsize $f_{ 1, 4 }$] {};
    \node[state] (f1) at (1,0) [label=above: \scriptsize $f_{ 1, 2 }$] {};
    \node[state] (f2) at (1,-1) [label=below: \scriptsize $f_{ 3, 2 }$] {};
    \node[state] (f3) at (0,-1) [label=below: \scriptsize $f_{ 3, 4 }$] {};
  \end{pgfonlayer}
  \begin{pgfonlayer}{behind}
    \draw
      (f0) -- node[above] {\scriptsize $ x_{1} $}  (f1) 
        -- node[right] {\scriptsize $ x_{2} $} (f2) 
        -- node[below] {\scriptsize $ x_{3} $} (f3)
        -- node[left] {\scriptsize $ x_{4} $} (f0);
  \end{pgfonlayer}
\end{tikzpicture}
		\caption{The S-NFG $\Nfcyc$.} \label{Sec:3:fig:1}	
	\end{minipage}
	\begin{minipage}[b]{0.15\textwidth}
		\centering
		\hspace*{-0.8cm}
		\begin{tikzpicture}[node distance=1.4cm, on grid,auto]
  \tikzstyle{state}=[shape=rectangle, draw, minimum width=0.4cm, minimum height = 0.4cm,outer sep=-0.3pt]
  \begin{pgfonlayer}{main}
    \node (f0) at (0,0) [] {};
    \node[state] (f1) at (1,0)
    [label=above: \tiny $ M_{ X_{4} | X_{1} } $] {};
    \node[state] (f2) at (2,0)
    [label=above: \tiny 
    $ M_{X_{1}\mathrm{,}X_{2}} $] {};
    \node[state] (f3) at (3,0)
    [label=above: \tiny $ M_{ X_{3} | X_{2} } $] {};
    \node (f4) at (4,0) [] {};
  \end{pgfonlayer}
  \begin{pgfonlayer}{behind}
    \draw
      (f0) -- node[above] {\scriptsize $ x_{4} $}  (f1) 
        -- node[above] {\scriptsize $ x_{1} $} (f2) 
        -- node[above] {\scriptsize $ x_{2} $} (f3)
        -- node[above] {\scriptsize $ x_{3} $} (f4);
  \end{pgfonlayer}
\end{tikzpicture}
		\caption{The S-NFG $ \NMkov $.}
		\label{Sec0:fig:1}
	\end{minipage}
	\begin{minipage}[b]{0.15\textwidth}
	    \centering
	    \begin{tikzpicture}[node distance=1.5cm, on grid,auto]
  \tikzstyle{state}=[shape=rectangle, draw, minimum width=0.4cm, minimum height = 0.4cm,outer sep=-0.3pt]
  \begin{pgfonlayer}{main}
    \node[state] (f0) at (0,0) [label=above: \scriptsize $f$] {};
    \node (f1) at (0.8,0.8) [] {};
    \node (f2) at (0.8,-0.8) [] {};
    \node (f3) at (-0.8,-0.8) [] {};
    \node (f4) at (-0.8,0.8) [] {};
  \end{pgfonlayer}
  \begin{pgfonlayer}{behind}  
    \foreach \x in {1,4}{
      \draw (f0) --node[above]{\scriptsize$x_{\x}$} (f\x);
    }
    \foreach \x in {2,3}{
      \draw (f0) --node[below]{\scriptsize$x_{\x}$} (f\x);
    }
  \end{pgfonlayer}
\end{tikzpicture}
	    \caption{The S-NFG $ \SNFG $.\label{Sec1:fig:10}}
  	\end{minipage}

	\vspace*{0.2cm}
    \begin{minipage}[t]{0.22\textwidth}
    	\centering
    	\hspace*{-0.6cm}
      	\begin{tikzpicture}[node distance=0.6cm, on grid,auto]
	\input{Figures/Fig_4A_Head};
	\input{Figures/Tikz_rho};
  	\begin{pgfonlayer}{main}
	  	\foreach \x/\xpm/\xin in {0/-1/1, 1/1/2}{
		  	\node[state1] (U\x) at ( \disbox, 
		  	3 * \scale - \x * 4 * \scale  ) {\scriptsize$U_{\xin}$};
		  	\node[state1] (UH\x) at ( \disbox, 
		  	\scale - \x * 4 * \scale ) {\scriptsize$U_{\xin}^{\herm}$};	
		  	 \node [dot=black] at (U\x.east) {};
		    \node [dot=black] at (UH\x.west) {};
		  	\foreach \y/\ypm in {1/-1,3/1}{
		        \ifthenelse{\x = 0}{
		            \pgfmathtruncatemacro\yin{\y} 
		        }{
		            \pgfmathtruncatemacro\yin{\y+1}
		        }
		        \node (x\x\y) at  ( \y*\disbox/2 +\ypm*\Ws/2-\ypm*\Wss/2, 
		        3 * \scale - \x * 4 * \scale + 2*\vergap ) [] 
		        {\scriptsize$ x_{\yin} $}; 
		        \node (xp\x\y) at ( \y*\disbox/2 +\ypm*\Ws/2-\ypm*\Wss/2, 
		        \scale - \x * 4 * \scale + 2*\vergap ) [] 
		        {\scriptsize$ x_{\yin}' $};
		    }
		    \node (px\x) at  ( 1.5*\disbox, - \xpm * 3.75*\Ws  ) [] {};
	  	}
  	\end{pgfonlayer}
  	\input{Figures/Fig_4A_basic};
\end{tikzpicture}
      	\caption{The Q-NFG $ \NQFG $.\label{Sec1:fig:1}}
    \end{minipage}
    \begin{minipage}[t]{0.22\textwidth}
    	\centering
    	\hspace*{-0.9cm}
      	\begin{tikzpicture}[node distance=0.6cm, on grid,auto]
	    	\input{Figures/Fig_7.tex}
		    \input{Figures/Tikz_rho};
		    \input{Figures/Fig_4A_basic};
		    \tikzstyle{state_dash1}=[shape=rectangle, draw, dashed, minimum width= 2.8*\disbox cm, minimum height = 7.4*\scale cm,outer sep=-0.3pt]
		    \begin{pgfonlayer}{behind}
		        \node[state_dash1] (dash1) at (\disbox, 0) [] {};
		        \node (exp0) at  ( -0.3*\scale, 5.9*\Ws ) []  
		        {\scriptsize$ \beli_{i,j}(x_{i}, x_{j}) $}; 
		    \end{pgfonlayer}
	    \end{tikzpicture}
      	\caption{The Q-NFG representation of 
      	$ \beli_{i,j}(x_{i},x_{j}) $.\label{Sec1:fig:7}}
    \end{minipage}
	
	\vspace*{0.2cm}
  	\begin{minipage}[t]{0.45\textwidth}
  		\centering
  		\vspace*{0.1cm}
  		\begin{tikzpicture}[]
  			\input{Figures/Fig_6.tex}
		  	\input{Figures/Fig_5.tex}
		\end{tikzpicture}
		\caption{The Venn diagram for $ \MP(\Nfcyc) $, $ \MP(\NMkov) $, $ \MP(\set{K}) $, $ \MP(\NQFG) $, and $ \setL $. \label{Sec1:fig:4} } 
	\end{minipage}
\end{figure}

To better understand classicable variables' marginals, we define the set
$ \MP(\NQFG) $, which is the set of the marginals created by the classicable
variables in the two-qubit system $ \NQFG $, as shown in Fig.~\ref{Sec1:fig:1}.  One of our paper's main topics is to fully characterize $ \MP(\NQFG) $. For comparison, we introduce $ \setL $ (the local
marginal polytope of the S-FG $ \Nfcyc $ in Fig.~\ref{Sec:3:fig:1}),
$ \MP(\Nfcyc) $ (the set of realizable
marginals of $ \Nfcyc $), $ \MP( \NMkov ) $ (the set of realizable
marginals of the Markov chain $ \NMkov $ in Fig.~\ref{Sec0:fig:1}), and $ \MP(\SNFG) $ (the set of realizable marginals of $ \SNFG $ in Fig.~\ref{Sec1:fig:10}). We have the following results.

\begin{itemize} 
	 
  \item We prove the Venn diagram in Fig.~\ref{Sec1:fig:4} by showing that each part in the diagram is non-empty. We can see that $ \MP(\SNFG) $  and $ \MP(\NQFG) $ are strict subsets of
  $ \setL $; both $ \MP( \Nfcyc ) $ and $ \MP(\NMkov) $ have
  marginals that are not in $ \MP(\NQFG) $; the set $ \MP(\NQFG) $ consists of
  marginals that are not compatible with any joint PMF.

  \item We generalize the Clauser-Horne-Shimony-Holt (CHSH) inequality~\cite{PhysRevLett.23.880} for Pearson correlation coefficients (PCCs), which resolves a conjecture proposed in~\cite{Victor2017}. Because PCCs are non-linear functions with respect to marginals, the inequality has a non-trivial proof.
  We suspect that the proof approach is applicable for proving other non-linear Bell inequalities. A violation of this inequality indicates that the associated marginals are not in $\MP(\SNFG) $. 

  \item We illustrate Hardy's paradox, Bell's game, and the maximum quantum violation of the PCC-based CHSH inequality by the classicable variables in $ \NQFG $ in Fig.~\ref{Sec1:fig:1}.
 
\end{itemize}

Besides these specific results, our paper is, more generally, about
leveraging tools from factor graphs to understand certain quantities of
interest in quantum information processing. In particular, given that factor
graphs have been proven very useful in classical information processing, but can also be used for doing quantum information processing, they allow one to
understand and appreciate the similarities and the differences between
classical and quantum information processing.

The rest of this paper is structured as follows. Section~\ref{Sec:2} reviews
some basics of S-FGs. In particular, Section~\ref{Sec:7} proves the
PCC-based CHSH inequality, and Section~\ref{Sec:8} discusses the Markov chain in Fig.~\ref{Sec0:fig:1}. Section~\ref{Sec:3} introduces $ \NQFG $, proves the Venn diagram in Fig.~\ref{Sec1:fig:4} and illustrates the maximum quantum violation of the PCC-based CHSH inequality. Many details are left out due to space constraints; a
more detailed discussion is given in~\cite{Huang2021}.

\subsection{Basic Notations and Definitions}


The sets 
$\mathbb{R}$, $\mathbb{R}_{\geq 0}$, $\mathbb{R}_{>0}$, and $\mathbb{C}$ denote the field of real numbers, the set of nonnegative real numbers, the set of positive real numbers, and the field of complex numbers,
respectively. An overline denotes complex conjugation. For any statement $ S $, by the Iverson's convention, the function $ [S] $ is defined to be $ [ S ] := 1 $ if $ S $ is true and $ [ S ] := 0 $ otherwise. 


\section{Standard Normal Factor Graphs (S-NFGs)~\label{Sec:2}}


In this section, we review some basic concepts and properties of an S-NFG. The word ``normal'' refers to the fact that variables are arguments of only one or two local functions.  We use an example to introduce the fundamental concepts of an S-NFG first.


\begin{example}~\cite{Kschischang2001, Loeliger2004a} 
  Consider the multivariate function
    \begin{align*}
      & g_{\Nfcyc}(x_{1},\ldots,x_{4}) \\ 
      &\quad : = f_{1,2}(x_{1}, x_{2}) 
      \cdot f_{1,4}(x_{1}, x_{4}) \cdot f_{3,2}(x_{3}, x_{2}) 
      \cdot f_{3,4}(x_{3}, x_{4}),
    \end{align*}
    where $g_{\Nfcyc}$, the
    so-called global function, is defined to be the product of the so-called
    local functions $f_{1,2}$, $f_{1,4}$, $f_{3,2}$ and $f_{3,4}$. We can visualize the
    factorization of $g$ with the help of the S-FG $\Nfcyc$
    in Fig.~\ref{Sec:3:fig:1}. Note that the S-FG $\Nfcyc$ 
    consists of four function nodes
    $f_{1,2}, \ldots, f_{3,4}$ and four (full) edges with associated variables $x_{1}, \ldots, x_{4}$.
\end{example}

For an S-NFG, a half edge is an edge incident on one function node only and a full edge is an edge incident on two function nodes. 

\begin{definition}
  The S-NFG $\mathsf{N}(\setF( \mathsf{N} ), \setE( \mathsf{N} ), \setx( \mathsf{N} ))$ consists of:
  \begin{enumerate}
    
  \item The graph $( \setF( \mathsf{N} ),\setE( \mathsf{N} ) )$ with vertex set $\setF( \mathsf{N} )$ and edge set $\setE( \mathsf{N} )$, where
   $\setE( \mathsf{N} )$ consists of all full edges and half edges in
   $\mathsf{N}$. With some slight abuse of notation, an $f \in \setF( \mathsf{N} )$ will denote a function node and the corresponding local function.

  \item The alphabet
    $\setx( \mathsf{N} ) :=\prod_{e\in \setE( \mathsf{N} )}\setx_{e}$,
    where $\setx_{e}$ is the alphabet associated with the edge
    $e\in \setE( \mathsf{N} )$. 

  \end{enumerate}
\end{definition}

\begin{definition}\label{Sec0:def:2}	
	Given $ \mathsf{N}(\setF( \mathsf{N} ), \setE( \mathsf{N} ), \setx( \mathsf{N} )) $, we make the following definitions:
  	\begin{enumerate}

	  	\item For every function node $f\in \setF( \mathsf{N} )$, the set $\partial f$ is the set of edges incident on $f$.


	  	\item An assignment
	    $\vx:=(x_{e})_{e\in \setE( \mathsf{N} )}\in \setx( \mathsf{N} )$ is
	    called a configuration of the S-NFG. 

	  	\item The local function $f$ associated with function node
	    $f\in \setF( \mathsf{N} )$ denotes an arbitrary mapping $  f: \prod_{e \in \partial f} \setx_{e} \to \mathbb{R}_{\geq 0}. $
	   
	    \item The global function is $ g_{ \mathsf{N} }(\vx) :=
		 	\prod_{ f\in \setF(\mathsf{N}) } 
		 	f( \vxf ). $

		\item The partition function is $ Z( \mathsf{N} ) := \sum_{ \vx } g_{ \mathsf{N} }(\vx), $
		where $ \sum_{ \vx } $ denotes $ \sum_{ \vx \in \setx( \mathsf{N} ) } $.

		\item The PMF induced on $ \mathsf{N} $ is 
		 $ \pn( \vx ):= g_{ \mathsf{N} }(\vx)/ Z( \mathsf{N} ) . $

		\item Let $\set{I}$ be a subset of $\setE( \mathsf{N} )$ and let
      $\set{I}^{\mathrm{c}} := \setE( \mathsf{N} ) \setminus \set{I}$ be its
      complement. The marginal $ p_{\mathsf{N},\set{I}}( \vx_{\set{I}}) $ is defined to be $ p_{\mathsf{N},\set{I}}( \vx_{\set{I}}) 
      := \sum_{ \vxIc  } 
      \pn( \vx )$.
    \end{enumerate}
\end{definition}

\begin{definition}
  Considering $ \mathsf{N} \in \{ \Nfcyc, \NMkov, \SNFG \} $, we make the following definitions:
  \begin{enumerate}
    \item The alphabet $ \setxe $ is $ \setxe := \{0, 1\} $ for all $ e \in \setE( \mathsf{N} ) $.

    \item The set $ \set{K} $ is
                  $ \set{K} := \{ \{1,2\}, \{1,4\}, \{ 2,3 \}, \{ 3,4 \} \} $.

    \item For $ \{ i, j \} \in \set{K} $, the marginal $ \vpf $  is defined to be a $ | \setxe | $-by-$ | \setxe | $ matrix with the entry $ \vpf( x_{i}, x_{j} ) $ and the marginal $ \vpe $ is defined to be a $ | \setxe | $-by-$ | \setxe | $ diagonal matrix with $ \vpe( x_{i} ) $ being the $ x_{i} $-th diagonal term.
    

    \item The collection of matrices $ \vbeli $ is defined to be $ \vbeli := \bigl( ( \vbeli_{i,j} )_{\{ i, j \} \in \set{K}}, ( \vbeli_{i} )_{i\in \setvar} \bigr) $. In particular, the matrix $ \vbeli_{i,j} $ is defined to be a $ | \setxe | $-by-$ | \setxe | $ matrix with entry $ \beli_{i,j}( x_{i}, x_{j} ) \in \sR_{\geq 0} $ and the matrix $ \vbeli_{i} $ is defined to be a $ | \setxe | $-by-$ | \setxe | $ diagonal matrix with $ \beli_{i}( x_{i} ) \in \sR_{\geq 0} $ being the $ x_{i} $-th diagonal term.


    \item The set of realizable marginals of $ \mathsf{N} $ is defined to be
    \begin{align*}
      \!\!\!\!\!\!\!\!
      \MP(\mathsf{N}) :=
      \left\{ 
        \vbeli 
        \left| 
        \begin{array}{l}
          \text{there exists an $ \setF(\mathsf{N}) $ 
          such that}\\
          \vbeli_{i,j} = \vpf, \ \vbeli_{i} = \vpe, \
          \{ i,j \} \in \set{K}
        \end{array}
        \right.
      \right\}. 
    \end{align*}

    \item The set $ \setL $ is defined to be 
    \begin{align*}
       \!\!\!\!\!\setL:= 
       \left\{ 
          \vbeli
          \left|
          \begin{array}{l}
          0 \leq \beli_{i,j}( x_{i}, x_{j} ) \leq 1,\
          \forall x_{i}, x_{j}, i,j \\ 
          \sum_{ x_{j} }
          \beli_{i,j}( x_{i}, x_{j} ) = \beli_{i}( x_{i} ),\
          \forall x_{i}, i
          \\
          \sum_{ x_{i} }
          \beli_{i,j}( x_{i}, x_{j} ) = \beli_{j}( x_{j} ),\ 
          \forall x_{j}, j\\ 
          \sum_{ x_{i} } \beli_{i}( x_{i} ) = 1,\ \forall i
          \end{array}
          \right.
      \right\}.
    \end{align*}
    The set $ \setL $ is essentially the local marginal polytope
                of the S-NFG $ \Nfcyc $ in Fig.~\ref{Sec:3:fig:1}. The
                definition of the local marginal polytope for an S-NFG is given
                in~\cite[Section 4.1.1]{Wainwright2008}.

    \item For each $ \vbeli \in \setL $ and $ \{ i,j \} \in \set{K} $, each marginal $ \vbeli_{i,j} $ can be used to represent the PMF for two random variables $ Y_{1}, Y_{2} \in \setxe $ by setting the probability $ \pr( Y_{1} = x_{i}, Y_{2} = x_{j} )
        = \beli_{i,j}(x_{i}, x_{j}) $, $ x_{i}, x_{j} \in \setxe $.
    The functions $ \cov( Y_{1}, Y_{2} ) $, $ \vari( Y_{1} ) $, $ \vari( Y_{2} ) $ are defined to be the covariance of $ Y_{1} $ and $ Y_{2} $, and the variances of $ Y_{1} $ and $ Y_{2} $, respectively.  When $\vari( Y_{1} ), \vari( Y_{2} ) >  0$, the PCC of $ Y_{1} $ and $ Y_{2} $ is defined to be
    $ \corr( \vbeli_{i,j} ) := \cov( Y_{1}, Y_{2} ) 
    / \sqrt{ \vari( Y_{1} ) \cdot \vari( Y_{2} ) }. $
  \end{enumerate}
\end{definition}

When there is no ambiguity, we use short-hands $ ( \cdot )_{i,j} $, $ ( \cdot )_{i} $, 
$ \sum_{x_{i}} $, and $\{ \cdot \}_{x_{i}} $  for $ ( \cdot )_{\{ i,j \} \in \set{K} } $, $ ( \cdot )_{i\in \setvar} $, $ \sum_{x_{i}\in \setxe} $, and 
$\{ \cdot \}_{x_{i}\in \setxe} $, respectively.

Because $ \setL $ is a convex set by definition, Carathéodory's theorem~\cite[Proposition B.6]{bertsekas2016nonlinear} states that each element in $ \setL $ can be written as a convex combination of the vertices in $ \setL $. The full list of the vertices in $ \setL $ is given in~\cite[Appendix~A]{Huang2021}.

\begin{proposition}
  For $ \{ i,j \} \in \set{K} $ and $ 0 < \beli_{i}(0), \beli_{j}(0) < 1 $, the PCC $ \corr( \vbeli_{i,j} ) $ satisfies
    \begin{align*}
        \corr( \vbeli_{i,j} ) &= \frac{ \det( \vbeli_{i,j} ) 
        }{ \sqrt{ \det( \vbeli_{i} ) \cdot \det( \vbeli_{j} ) 
        } }.
    \end{align*}
    The requirement $ 0 < \beli_{i}(0), \beli_{j}(0) < 1 $ ensures that $ \det( \vbeli_{i} ), \det( \vbeli_{j} ) > 0 $, and thus $ \corr( \vbeli_{i,j} ) $ is well-defined.
\end{proposition}
\begin{proof}
  See the proof of~\cite[Corollary~9]{Huang2021}.
\end{proof}

\begin{definition}\label{Sec0:def:1}
  Suppose that $\vbeli \in \setL $ and $ 0 < \beli_{i}(0) < 1,\ i \in \setvar $, we define
  \begin{align*}
    &\rCHSHbet :=   \nonumber\\
    &\quad\ \corr( \vbeli_{1,2} ) + \corr( \vbeli_{1,4} )
    + \corr( \vbeli_{3,2} )- \corr( \vbeli_{3,4} ).
  \end{align*}

\end{definition}

\subsection{Properties for \texorpdfstring{$\SNFG$}{}\label{Sec:7}}

In this subsection, we prove inequalities with respect to $ \rCHSHbet $ for $ \vbeli \in  \MP(\SNFG) $. These inequalities genuinely are (nonlinear) Bell inequalities~\cite{PhysRevLett.48.291} in the usual sense. By definition, it holds that
\begin{align*}
  \MP(\Nfcyc) \subseteq \MP(\SNFG), \qquad
  \MP(\NMkov) \subseteq \MP(\SNFG),
\end{align*}
so any inequality that holds for all $ \vbeli \in \MP(\SNFG) $ also holds for all $ \vbeli \in \MP(\Nfcyc) \cup \MP(\NMkov) $.

\begin{theorem}\label{Sec0:prop:3}
  For any $ \vbeli \in \MP(\SNFG) $ such that $ 0 < \beli_{i}(0) < 1 $ for all
  $ i \in \setvar $, we have 
  \begin{align*}
    |\rCHSHbet| 
      &< 2\sqrt{2}.
  \end{align*}
\end{theorem}
\begin{proof}
  We prove it by contradiction. On the one hand, the set $ \MP( \SNFG ) $ consists of marginals for binary random variables only. On the other hand, to have $ \rCHSHbet = 2 \sqrt{2} $ for some $ \vbeli \in \MP(\SNFG) $, the PMF realizing $ \vbeli $ needs to be the joint PMF for random variables with alphabet size greater than two. For details, see the proof in~\cite[Appendix~C]{Huang2021}.
\end{proof}

The main idea in the proof of Theorem~\ref{Sec0:prop:3} can be used to verify whether a proposed bound for a function with binary random variables is achievable. It is different from the idea in the proof of the upcoming Theorem~\ref{Sec0:thm:2}. 

\begin{theorem}\label{Sec0:thm:2}
  For any $ \vbeli \in \MP(\SNFG) $ such that $ 0 < \beli_{i}(0) < 1 $ for all $ i \in  \setvar $, we have $ |\rCHSHbet| \leq 5/2. $
\end{theorem}
\begin{proof}
  We give a proof sketch here. 
  For details, see the proof in~\cite[Appendix~E]{Huang2021}.
  \begin{itemize} 

    \item Consider a subset of $ \setL $ such that in this subset, $ 0 < \beli_{i}(0) < 1 $ for $ i \in \setvar $, and the elements in $ \vbeli $ satisfy the original linear CHSH inequality. Denote this set as $ \setLCHSH $. We have $ \MP(\SNFG) \subsetneq \setLCHSH $.

    \item Find a $\vbeli^{*} \in \MP(\SNFG)$ such that $ \mathrm{CorrCHSH}(\vbeli^{*}) = 5/2. $
   
    \item We formulate an optimization problem where $ \rCHSHbet $ is maximized over $ \vbeli \in \setLCHSH $ such that $ \vbeli $ has a similar structure as $ \vbeli^{*} $, e.g., having the same number of zero entries in $ (\vbeli_{i,j})_{i,j} $. Note that this optimization problem has linear constraints only, which helps determine the optimal solution. We prove $ \rCHSHbet \leq 5/2 $ in this case.

    \item We generalize the proof for all $ \vbeli \in \setLCHSH $.

    \item The proof of $ \rCHSHbet \geq -5/2 $ is similar.

   \end{itemize}
\end{proof}
Theorem~\ref{Sec0:thm:2} proves the conjecture stated in~\cite{Victor2017}. The key idea of the proof is that we consider $ \setLCHSH $ instead of $ \MP(\SNFG) $. Suppose that we want to prove $ \rCHSHbet \leq 5/2 $ for $ \vbeli \in \MP(\SNFG) $ directly. Because $ \MP(\SNFG) $ is a convex set, for any $ \vbeli \in \MP(\SNFG) $, the marginal $ \vbeli_{i,j} $ can be written as a convex combination of some joint PMF for $ X_{1},\ldots,X_{4} $, i.e., $ \{ p_{\SNFG}( \vx ) \}_{\vx} $, which makes the expression of $ \rCHSHbet $ non-trivial. By considering a superset of $ \MP(\SNFG) $, i.e., $ \setLCHSH $, we can simplify $ \rCHSHbet $. We suspect that this idea can be generalized in the proof of other non-linear Bell inequalities.

\subsection{Markov Chain in Fig.~\ref{Sec0:fig:1}}\label{Sec:8}

In this subsection, we consider the Markov chain $ \NMkov $ in Fig.~\ref{Sec0:fig:1}.

\begin{theorem}\label{Sec0:thm:1}
 For the Markov chain $ \NMkov $ in Fig.~\ref{Sec0:fig:1}, we have
  \begin{align*}
    \corr( \vbeli_{3,4} ) = \corr( \vbeli_{3,2} ) 
    \cdot \corr( \vbeli_{1,2} ) 
    \cdot \corr( \vbeli_{1,4} ).
  \end{align*}
\end{theorem}

\begin{proof}
  See~\cite[Corollary 19]{Mori2015}.
\end{proof}

\begin{corollary}\label{Sec1:coro:3}
  For the Markov chain $ \NMkov $ in Fig.~\ref{Sec0:fig:1}, it holds that $ | \corr( \vbeli_{3,4} ) | \leq | \corr( \vbeli_{1,2} ) | \leq 1. $
\end{corollary}
\begin{proof}
  It can be proven using Theorem~\ref{Sec0:thm:1} and $ |\corr( \vbeli_{i,j} )| \leq 1 $ for $ \{ i,j \} \in \set{K} $.
\end{proof}

We prove another variation of the PCC-based CHSH inequality for $ \NMkov $.
\begin{corollary}
  For the Markov chain $ \NMkov $ in Fig.~\ref{Sec0:fig:1}, we have
  \begin{align*}
    \bigl| & \corr( \vbeli_{1,2} ) + \corr( \vbeli_{2,4} ) 
    + \corr( \vbeli_{1,3} ) - \corr( \vbeli_{3,4} ) \bigr| 
    \leq 2.
  \end{align*}
\end{corollary}
\begin{proof}
  See the proof of~\cite[Proposition~20]{Huang2021}.
\end{proof}


\section{Quantum-Probability Normal Factor Graphs (Q-NFGs)\label{Sec:3}}


This section considers a quantum system represented by the Q-NFG $ \NQFG $ in
Fig.~\ref{Sec1:fig:1}. Such Q-NFGs have been discussed thoroughly
in~\cite{Loeliger2017, Loeliger2020}. Note that in Fig.~\ref{Sec1:fig:1} and
Fig.~\ref{Sec1:fig:7}, the row index of a matrix is marked by a black dot. The
details of $ \NQFG $ are shown in \cite[Definition~21]{Huang2021}.

\begin{proposition}\label{Sec:3:prop:4}
 	For any $ \{ i, j \} \in \set{K} $, the variables $ \tx_{i} $ and $ \tx_{j} $ are jointly classicable, which implies that the marginals $ q_{i,j}( \tx_{i}, \tx_{j} ) $ and $ q_{i}( \tx_{i} ) $ are non-negative real numbers for any $ \tx_{i}, \tx_{j} \in \setxe^{2} $.
\end{proposition}
\begin{proof}
 	It can be proven directly by Definition~\ref{Sec0:def:3}.
\end{proof}

Then we define the set of realizable marginals of $ \NQFG $ based on the jointly classicable variables $ \tx_{i} $ and $ \tx_{j} $ 
for all $ \{ i, j \} \in \set{K} $.

\begin{definition}
 	With $ \tzero := (0,0) $, $ \tone := (1,1) $, and $ \{ i, j \} \in \set{K} $, the matrices $ \vq_{i,j} $ 
 	and $ \vq_{i} $ induced by $ q_{\NQFG} $ are defined to be
	\begin{align*}
	 	\vq_{i,j} &:= 
	 	\begin{pmatrix} 
	 		q_{i,j}( \tzero, \tzero ) 
	 		& q_{i,j}( \tzero, \tone ) \\
	 		q_{i,j}( \tone, \tzero )
	 		& q_{i,j}( \tone, \tone )
	 	\end{pmatrix}, \
	 	\vq_{i} := 
	 	\begin{pmatrix}
	 		q_{i}( \tzero ) & 0 \\
	 		0 & q_{i}( \tone )
	 	\end{pmatrix}.\nonumber
	\end{align*}
	The set of realizable marginals of $ \NQFG $ is defined to be the set 
	$ \MP(\NQFG) :=
 	\left\{ 
 		\vbeli 
		\left| 
		\vbeli_{i,j} = \vq_{i,j}, \ \vbeli_{i} = \vq_{i},\
		\{ i, j \} \in \set{K}
		\right. 	 	
	\right\}$,
  which is not the set of quantum correlations in the usual Bell nonlocality sense.
\end{definition}

\begin{proposition}\label{Sec:3:prop:6}
 	For any $ \vbeli \in \MP(\NQFG) $, there exist matrices $ \rho $, $ U_{1} $, and $ U_{2} $ such that
	\begin{align}
	 	\!\!\!\beli_{i,j}(x_{i}, x_{j}) &= 
    \Tr( ( A_{i,x_{i}} \otimes B_{j,x_{j}} )
    \cdot \rho \cdot ( A_{i,x_{i}} \otimes B_{j,x_{j}} )^{\herm} ),
	 	\label{Sec:3:eqn:15}
	\end{align}
	for all $ x_{i}, x_{j} \in \setxe$ and $\{ i,j \} \in \set{K}$, where 
	\begin{align*}
	  &A_{i,x_{i}} := E_{x_{i}} \cdot U_{1}^{[i = 3]},\
    B_{j,x_{j}} := E_{x_{j}} \cdot U_{2}^{[j = 4]},  \\
	 	&E_{x_{i}}( y_{i}, y_{i}' ) := 
	 	[ y_{i} = x_{i}] \cdot [y_{i} = y_{i}' ], \ 
    x_{i}, y_{i}, y_{i}' \in \setxe.
	\end{align*}
	Note that the set $ \{ E_{x_{i}} \}_{x_{i}} $ denotes the measurement of a single qubit in the computational basis. Then we have
	\begin{align*}
		\sum_{x_{i} \in \setxe} A_{i,x_{i}}^{\herm} 
    \cdot A_{i,x_{i}} &= 
    \sum_{x_{j} \in \setxe} B_{j,x_{j}}^{\herm} 
    \cdot B_{j,x_{j}} = I, \quad \{ i, j \} \in \set{K},
	\end{align*}
	which shows that both $ \{ A_{i,x_{i}} \}_{x_{i}} $ and $ \{ B_{j,x_{j}} \}_{x_{j}} $ are sets of measurement matrices with binary outcomes $ x_{i} $ and $ x_{j} $, respectively.
\end{proposition}
\begin{proof}
  It can be proven directly.
\end{proof}

After closing the dashed box in Fig.~\ref{Sec1:fig:7}, i.e., summing over the variables inside the box, we obtain~\eqref{Sec:3:eqn:15}.

\begin{proposition}\label{Sec:3:prop:5}
  There exists a $ \vbeli \in \MP(\NQFG) $ such that $ |\corr( \vbeli_{3,4} )| > |\corr( \vbeli_{1,2} )| $.
\end{proposition}
\begin{proof}
  See~\cite[Proposition~26]{Huang2021}.
\end{proof}

Compared with Corollary~\ref{Sec1:coro:3}, Proposition~\ref{Sec:3:prop:5} implies that $ \MP(\NQFG) $ provides extra $ \vbeli $ that is not in $ \MP(\NMkov) $.

\begin{theorem}\label{Sec1:thm:1}
 	The Venn diagram in Fig.~\ref{Sec1:fig:4} holds.
\end{theorem}
\begin{proof}
  See the proof of~\cite[Theorem~49]{Huang2021}.
\end{proof}

We make some remarks on the Venn diagram in Fig.~\ref{Sec1:fig:4}:

\begin{itemize}

	\item On the one hand, the set $ \MP(\NQFG) $ provides extra marginals that are not in $ \MP( \SNFG ) $. For example, by introducing entanglement in the quantum system, one can obtain a set of incompatible marginals (see, e.g.,~\cite[Example~30]{Huang2021}). 

	\item On the other hand, the sets $ \MP(\Nfcyc) $, $ \MP(\NMkov) $, and $ \MP( \SNFG ) $ also consist of marginals that are not in $ \MP(\NQFG) $. 

\end{itemize}

\begin{proposition}\label{Sec:3:prop:2}
 	For $ \vbeli \in \MP(\NQFG) $ satisfying $ 0 < \beli_{i}(0) < 1 $ for all $ i \in  \setvar $, we have 
 	$ |\rCHSHbet|  \leq 2\sqrt{2} $.
\end{proposition}
\begin{proof}
 	See~\cite[Appendix B]{Victor2017}. 
\end{proof}

\begin{proposition}
   Hardy's paradox~\cite{Hardy1992} and Bell's game~\cite{Gisin2014} can be illustrated via the classicable variables induced by the SQMF $ q_{\NQFG}( \tvx ) $. In Bell's game, we have $ \rCHSHbet = 2\sqrt{2} $, which also realizes the maximum quantum violation of $ \rCHSHbet $ as proven in Proposition~\ref{Sec:3:prop:2}.
\end{proposition}
\begin{proof}
  See~\cite[Example~30, Proposition~31]{Huang2021}.
\end{proof}

\clearpage 

\balance
\begin{footnotesize}
 \bibliography{biblio}
\end{footnotesize}

\end{document}